\documentclass[11pt]{article}

\usepackage[utf8]{inputenc}
\usepackage{graphicx}
\usepackage{amsmath,amssymb,amsthm,appendix}
\usepackage{multicol, multirow}
\usepackage{bigstrut}
\usepackage{adjustbox,lipsum}
\usepackage{geometry,setspace}
\usepackage{natbib}
\usepackage{algpseudocode}
\usepackage{algorithm}
\usepackage{relsize}
\usepackage{bm}
\usepackage{xcolor}
\usepackage{authblk}

\newcommand{\indep}{\perp \!\!\! \perp}

\newtheorem{theorem}{Theorem}
\newtheorem{corollary}{Corollary}[theorem]
\providecommand{\keywords}[1]
{
  \small	
  \textbf{\textit{Keywords---}} #1
}

\begin{document}

\title{Multiple Imputation for Non-Monotone Missing Not at Random Binary Data using the No Self-Censoring Model}

\author[1,2,*]{Boyu Ren}
\author[4,5]{Stuart R. Lipsitz}
\author[3,4]{Roger D. Weiss}
\author[1,2]{Garrett M. Fitzmaurice}

\affil[1]{Laboratory for Psychiatric Biostatistics, McLean Hospital, Belmont, MA, USA}
\affil[2]{Department of Psychiatry, Harvard Medical School, Boston, MA, USA}
\affil[3]{Division of Alcohol and Drug Abuse, McLean Hospital, Belmont, MA, USA}
\affil[4]{Department of Medicine, Harvard Medical School, Boston, MA, USA}
\affil[5]{Department of Medicine, Harvard Medical School, Boston, MA, USA}
\affil[*]{Corresponding author: Boyu Ren, bren@partners.org}
\date{}
\setcounter{Maxaffil}{0}
\renewcommand\Affilfont{\itshape\small}

\maketitle

\begin{abstract}
Although approaches for handling missing data from longitudinal studies are well-developed when the patterns of missingness are monotone, fewer methods are available for non-monotone missingness. Moreover, the conventional missing at random (MAR) assumption---a natural benchmark for monotone missingness---does not model realistic beliefs about non-monotone missingness processes (Robins and Gill, 1997). This has provided the impetus for alternative non-monotone missing not at random (MNAR) mechanisms. The “no self-censoring” (NSC) model is such a mechanism and assumes the probability an outcome variable is missing is independent of its value when conditioning on all other possibly missing outcome variables and their missingness indicators. As an alternative to “weighting” methods that become computationally demanding with increasing number of outcome variables, we propose a multiple imputation approach under NSC.  We focus on the case of binary outcomes and present results of simulation and asymptotic studies to investigate the performance of the proposed imputation approach. We describe a related approach to sensitivity analysis to departure from NSC. Finally, we discuss the relationship between MAR and NSC and prove that one is not a special case of the other. The proposed methods are illustrated with application to a substance use disorder clinical trial. 
\end{abstract}


\section{Introduction}
Missing data are a common and challenging problem that complicates the analysis of data from research studies across a wide spectrum of fields. In longitudinal study designs the problem of missing data is far more acute than in cross-sectional study designs because missing data can arise at any of the follow-up occasions. It is crucial to properly account for missing data in the analysis stage to obtain unbiased estimates of the parameters of primary interest. Statistical approaches for handling missing data have been well developed when the patterns of missingness are monotone, as might arise in a longitudinal study with attrition or dropout, and when the missingness depends only on the observed data (i.e., when data are missing at random or MAR) \citep{robins1994estimation,tsiatis2006semiparametric,molenberghs2014handbook,little2019statistical}. For example, two widely used approaches for handling missingness are multiple imputation (MI) \citep{rubin1996multiple,little2002bayes} and inverse-probability weighting (IPW) \citep{robins1994estimation}. These two distinct approaches differ in the following way. MI specifies a model for the joint distribution of the missing data given the observed data; in contrast, IPW specifies a model for the probability that an individual has complete data or a specific pattern of observed data. Both approaches are relatively straightforward to implement when the missing data patterns are monotone and the missingness is assumed to be MAR. 

However, there are many longitudinal studies where the restriction to monotone missing data patterns does not hold and the MAR assumption is thought to be questionable. For example, in longitudinal studies of substance use disorders, missing data often occur intermittently, with data on an individual missing at one follow-up occasion and then measured at a later occasion. This leads to a series of complex non-monotone patterns of missingness in the data. Additionally, when the measurements at each study visit collect sensitive information, such as on drug use, this can strongly influence a patients’ decision about whether to appear or skip a study visit. As argued in \cite{robins1997non} and \cite{vansteelandt2007estimation}, this scenario often results in data that are missing not at random (MNAR), i.e., the probability of missingness may also depend on part of the unobserved data. Thus, 
although the MAR assumption is appealing and often considered a benchmark assumption for monotone missing data, it is very difficult to conceive of models for non-monotone missingness that do not impose more restrictive assumptions than what MAR strictly entails. That is, it is difficult to consider non-monotone MAR mechanisms that are not also missing completely at random (MCAR) \citep{robins1997non}. For the case of non-monotone missingness, \cite{linero2018bayesian} provide a very comprehensive review of many alternative identifying restrictions to MAR.

Recent developments for inference on non-monotone and/or MNAR data have mostly focused on applying inverse probability weighting (IPW) with certain identifying assumptions on the missingness mechanism \citep{rotnitzky1998semiparametric,robins2000sensitivity,vansteelandt2007estimation,zhou2010block,li2013weighting,shpitser2016consistent,sadinle2017itemwise,sun2018inverse,tchetgen2018discrete,malinsky2022semiparametric}. On the other hand, imputation approaches for non-monotone MNAR data are relatively scarce and often the theoretical underpinnings are somewhat unclear 
\citep{tompsett2018use,
fletcher2020missing}.
Recently, a model-based imputation approach for MNAR data with pre-specified assumptions on the missing data mechanism was proposed by \cite{beesley2021multiple}; however, 
their method relies on a strong assumption that 
the missingness indicators (binary variables that equal 0 if an outcome variable is observed and 1 if missing) are conditionally independent given the outcome variables.  
In this paper, we adopt the identifying assumption introduced in \cite{malinsky2022semiparametric} for IPW, named “no self-censoring” (NSC), but instead of using IPW, we construct a multiple imputation approach based on this assumption. Specifically, we assume that the probability an outcome variable is missing is independent of its value when conditioning on all other possibly missing outcome variables and their missingness indicators. Unlike \cite{beesley2021multiple}, the NSC assumption allows for very general dependence among the missingness indicators (conditional on the outcome variables). Moreover, although missingness can depend on unobserved values, and hence is MNAR,  \cite{sadinle2017itemwise} and \cite{malinsky2022semiparametric} have shown that any measurable function of the full data distribution can be identified based on the observed data under the NSC assumption.

\cite{malinsky2022semiparametric} proposed an augmented IPW (AIPW) estimator for the NSC model. This estimator attains a higher level of efficiency compared to previous estimators proposed for the same model. In addition, the estimator is an appealing alternative to conventional MI when the data are MNAR according to the NSC assumption. However,  
the approach poses computational challenges when the number of outcome variables is relatively high. Specifically, the IPW weights associated with different missingness patterns are estimated through a set of estimating equations whose dimension is $O(3^K)$, where $K$ is the number of outcome variables. Solving these estimating equations simultaneously becomes computationally challenging when there are more than five or six outcome variables. To address this limitation, we propose to use multiple imputations, under NSC, based on fully conditional specification (FCS) to fill in the missing values directly. FCS models target at the fully conditional distribution of an outcome variable; to incorporate the NSC assumption into the proposed imputation approach, the set of variables conditioned on in the FCS models includes all other outcome variables as well as their missingness indicators. As a result, for $K$ outcome variables, we need to specify only $K$ regression models. The exact form of the regression models depends on the type of outcome variables and, in principle, a suite of flexible machine learning algorithms can be used to potentially mitigate any misspecification in the fully conditional distributions (e.g., capturing complex interactions between the covariates).
In this paper, we focus of the case where the outcome variables are binary. Of note, \cite{malinsky2022semiparametric} use simulations to compare the performance of their AIPW estimator to conventional MAR multiple imputation by FCS, demonstrating that the latter performs poorly 
when data are generated under a NSC model; this result is as expected. However, in the same simulation settings used by \cite{malinsky2022semiparametric}, we demonstrate
that our proposed NSC approach to multiple imputation by FCS is unbiased (see Figure S1 in the Supplementary Materials). 

We also introduce a straightforward approach for sensitivity analysis to departures from the NSC assumption that can be viewed as an alternative to the Bayesian approach in \cite{sadinle2017itemwise}. Sensitivity analysis is an integral analytic step when missing data are considered to be MNAR. It quantifies the robustness of inferences about parameters of interest to departures from the assumed missing data mechanism. Specifically, we show that there is a natural set of sensitivity parameters that can be introduced into NSC which describe the associations between the missingness indicators and their corresponding outcome variables. A suitable range for these sensitivity parameters can usually be elicited from investigators with subject-matter expertise. In the Methods section, we show that the proposed sensitivity analysis is also computationally efficient. Global sensitivity analysis anchored at nonparametric identified models for MNAR and non-monotone data has also been discussed in \cite{vansteelandt2007estimation,scharfstein2022}.

In Section 2, we review the NSC assumption and the implication of this assumption on the fully conditional distribution of each outcome variable given the other outcome variables and their missingness indicators. We then describe our proposed algorithm for multiple imputation under the NSC assumption and the related approach to sensitivity analysis. We also discuss the relationship between MAR and NSC and prove that one is not a special case of the other. In Section 3, we present results of simulation and asymptotic studies to investigate the performance of the proposed imputation approach in various settings. In Section 4, we illustrate the proposed method with an application to data from a longitudinal clinical trial of patients with substance use disorder; in this clinical trial, missingness was non-monotone and  42\% of study participants were missing at least one of the six planned monthly drug use assessments.

\section{Methods}
We assume the data are i.i.d.\ samples of $(M, Y)$ from an underlying data generating mechanism, where $Y = (Y_1,\ldots, Y_K)$ is the full data vector for the $K$-dimensional multivariate outcome and $M = (M_1,\ldots,M_K)$ is the vector of missingness indicators, with $M_k = 1$ when $Y_k$ is missing and $M_k = 0$ when $Y_k$ is observed. This notation for the missingness indicators is natural and transparent; it also follows the notation used in \cite{little2019statistical}, \cite{sadinle2017itemwise} and in the landmark report by the National Research Council (NRC) \citep{national2010prevention}. We note that it is reversed from another common notation in the literature where indicators for observed and missing variables are assigned the values 1 and 0 respectively.
In this paper, we focus on the case where all $Y_k$ are binary. Later, in the Discussion section, we also consider the extension of our approach to the case where $Y_k$ is continuous, ordinal or a count. We use $M_{-k}$ and $Y_{-k}$ to denote the vector of $M$ and $Y$ without their $k$-th components respectively. The observed sub-vector of $Y$, as indicated by $M$, is denoted by $Y_{(M)}$ (i.e., the components of Y with $M_{k}=0)$ and the observed data are thus i.i.d.\ realizations of $(M, Y_{(M)})$. We are interested in inference on parameters that describe some features of the joint distribution of $Y$, and possibly also their dependence on always-observed covariates (we will introduce additional notation and algorithms for scenarios with covariates later in this section).

\subsection{No Self-Censoring Model}
The no self-censoring (NSC) model, introduced in  
\cite{shpitser2016consistent}, \cite{sadinle2017itemwise} and \cite{malinsky2022semiparametric}, assumes that the missingness indicator $M_k$ of an outcome variable $Y_k$ is independent of $Y_k$ given all other outcome variables and their missingness indicators:
\begin{equation}
    M_k \indep Y_k~|~M_{-k}, Y_{-k}.
    \label{eq:NSC-assumption}
\end{equation}
Note that although the missingness indicator $M_k$ is assumed to be conditionally independent of the value of $Y_k$, the NSC model does allow for dependence on all other possibly missing outcomes (and their missingness indicators); as a result, NSC is an MNAR mechanism. As described in \cite{sadinle2017itemwise}, when $Y$ contains only binary variables, the NSC model for $(M, Y)$ has a deep connection with conventional loglinear models.

We define notation for the loglinear representation of a generic NSC model. Let $\mathcal P(S)$ denote the power set of a set $S$, $\mathcal P^{+}(S) = \mathcal P(S)\setminus \emptyset$, and $[K]$ the set $\{1,\ldots,K\}$. For an arbitrary set $I\subset[K]$, let $Y_I = \prod_{i\in I} Y_{i}$ and $M_I = \prod_{i\in I} M_i$. The joint distribution of $(M,Y)$ under the NSC assumption for any $K>1$ can be expressed with the following loglinear model
\begin{equation}
p(M,Y;\lambda) \propto \exp\left(\sum_{I\in \mathcal P([K])}\sum_{J\in \mathcal P([K]\setminus I)}\lambda_{M_IY_J}M_IY_J\right).
    \label{eq:NSC-generic}
\end{equation}
Here the $\lambda$'s are the model parameters. Recall that the unrestricted joint distribution of $(M,Y)$ can be fully described by a saturated loglinear model, where terms corresponding to all main effects and interactions up to order $2K$ are present. The NSC assumption is equivalent to omitting all interaction terms from the saturated loglinear model which involve both an outcome variable $Y_k$ and its missingness indicator $M_k$ for at least one $k\in[K]$. To be more concrete, consider the case where $K = 3$. The loglinear model for $(M,Y)$ under NSC is
\begin{equation}
\begin{aligned}
p(M,Y;\lambda) \propto \exp&\left(\sum_k (\lambda_{M_k}M_K + \lambda_{Y_k}Y_k) + \sum_{k\neq l}(\lambda_{M_kM_{l}}M_kM_{l} + \lambda_{Y_kY_{l}}Y_kY_{l})\right.\\
+ &\sum_{k\neq l}\lambda_{M_kY_l}M_kY_l+ \lambda_{M_1M_{2}M_3}M_1M_{2}M_3 + \lambda_{Y_1Y_2Y_3}Y_1Y_2Y_3\\
+ &\left.\sum_{k\neq l\neq r}(\lambda_{M_kM_lY_r}M_kM_lY_r + \lambda_{Y_kY_lM_r}Y_kY_lM_r)\right).
\end{aligned}
\label{eq:loglinear-NSC}
\end{equation}

Based on results from \cite{sadinle2017itemwise}, the set of constraints implied by the NSC assumption guarantees that the joint distribution of $(M,Y)$ is identifiable from the observed data $(M, Y_{(M)})$. Of note, although MAR models are also identifiable based on the observed data, we want to emphasize that MAR is not a special case of NSC. More specifically, any further constraints on the $\lambda$'s in the loglinear representation of NSC models that make them MAR will indeed reduce the models to MCAR; see Section 2.4 for more details.

A critical implication of the loglinear representation of NSC models in (\ref{eq:NSC-generic}) and (\ref{eq:loglinear-NSC}) is that the conditional distribution of $Y_k$ given $Y_{-k}$ and $M$ follows a logistic regression model. Specifically, when $K = 3$, we have for any $k\neq l \neq r$
\begin{equation}
\begin{aligned}
    \text{logit}(p(Y_k=1|Y_{-k}, M)) = &\lambda_{Y_k} + \sum_{l\neq k} (\lambda_{M_lY_k}M_l+\lambda_{Y_lY_k}Y_l) + \sum_{l\neq r\neq k}\lambda_{Y_k Y_lM_r} Y_lM_r\\
    +& \lambda_{Y_kY_lY_r}Y_lY_r + \lambda_{Y_kM_lM_r}M_lM_r.
\end{aligned}
\label{eq:logistic-NSC}
\end{equation}
When $K>3$, similar results can be derived for the fully conditional distributions of $Y_k$ where higher order interactions between $Y$ and $M$ are present. From equation (\ref{eq:logistic-NSC}), it also follows that $p(Y_k|Y_{-k},M) = p(Y_k|Y_{-k}, M_{-k}) = p(Y_k|Y_{-k}, M_{-k}, M_k = 0)$, which suggests that we can recover $p(Y_k|Y_{-k}, M_{-k}, M_k = 1)$ using samples where $Y_k$ is observed, i.e., from $p(Y_k|Y_{-k}, M_{-k}, M_k = 0).$

\subsection{Imputation via FCS under NSC}
The logistic regression model property indicated by (\ref{eq:logistic-NSC}) provides a natural route to multiple imputation under NSC through FCS \citep{van2006fully,van2007multiple}. The FCS approach, also known as chained equations \citep{van2011mice}, imputes multivariate missing data using the estimated fully conditional distributions of each variable iteratively.  Conventional FCS is typically applied under the assumption that the data are MAR (referred to as FCS-MAR hereafter). Under MAR, the imputation of missing $Y_k$ relies on the estimated conditional distribution $p(Y_k|Y_{-k},M_k=0)$ derived from samples with $Y_k$ observed, since MAR implies $p(Y_k|Y_{-k},M_k=0) = p(Y_k|Y_{-k},M_k=1)=p(Y_k|Y_{-k})$. A full sweep of this process over $k=1,\ldots,K$ constitutes an iteration of imputation; usually, multiple iterations are performed before the final imputation is generated. Typically, the imputation is performed $T>1$ times with random starts, and the uncertainty of the estimates of parameters of interest is quantified with Rubin's rules \citep{rubin2004multiple} based on all $T$ imputations. Confidence intervals, as well as hypothesis tests, can be obtained in a similar fashion. Similar to Gibbs sampling, the main idea of FCS is to recover the joint distribution of $Y$ from all conditional distributions $p(Y_k|Y_{-k})$. However since $p(Y_k|Y_{-k})$ are not exact in many instances, FCS usually does not have the same convergence guarantees as Gibbs sampling when the number of iterations goes to infinity \citep{heckerman2000dependency,rubin2003nested}.

We first describe the FCS approach for multiple imputation under NSC (hereafter referred to as FCS-NSC) when no covariates are present. Based on (\ref{eq:logistic-NSC}), the conditional distribution of $Y_k$ depends on $M_{-k}$ in addition to $Y_{-k}$. Therefore, we can fit a logistic regression model with $Y_k$ as the outcome and $(Y_{-k},M_{-k})$ as covariates, on the subset of the data where $Y_k$ is observed. This yields an estimate of the conditional distribution $p(Y_k|M_{-k}, Y_{-k})$, which serves as the central component in an FCS algorithm. In addition, under NSC, the logistic model will always be correctly specified for $p(Y_k|Y_{-k}, M_{-k})$ as suggested by (\ref{eq:logistic-NSC}), provided all relevant main effects and two-way and higher-order interactions are included. This means that with large sample size, the FCS approach can proceed without the need to address potential model misspecification. However, in applications, the logistic regression model in (\ref{eq:logistic-NSC}) is usually not practical, especially when $K$ is large. As a result, we suggest imposing constraints on the higher-order interaction terms by assuming their corresponding regression coefficients are zero. We illustrate the proposed FCS approach in Algorithm \ref{alg:fcs-main}.

\begin{algorithm}[htp]
\caption{Multiple imputation with FCS based on NSC assumption (FCS-NSC)}
\begin{algorithmic}[1]
\State Let $T$ be the number of imputations;
\State Let $R$ be the number of iterations for each imputation;
\State Let $Y^{(t)}$ be the results from the $t$-th imputation;
\For{$t\gets 1$ to $T$} \Comment{Each imputation replicate}
\For{$k\gets 1$ to $K$}
    \For{$i\gets 1$ to $N$}
        \If{$Y_{i,k}=\text{NA}$}
        \State $Y^{(t)}_{i,k}\gets$ sample from the empirical distribution $\hat p(Y_{k}|M_k = 0)$; \Comment{Initial filling}
        \EndIf
    \EndFor
\EndFor

\For{$r\gets 1$ to $R$} \Comment{Iterations within imputation}
    \For{$k\gets 1$ to $K$}
    \State $\hat p_r(Y_k|Y_{-k},M_{-k},M_k=0)\gets$ Logistic regression with augmentation \citep{white2010avoiding} on current $(Y^{(t)}, M)$;
    \State Update the missing values of $Y^{(t)}_k$ by sampling from $\hat p_r(Y_k|Y_{-k}, M_{-k}, M_k = 0)$;
    \EndFor
\EndFor
\EndFor
\State \Return $Y^{(1)}, Y^{(2)},\ldots, Y^{(T)}$;
\end{algorithmic}
\label{alg:fcs-main}
\end{algorithm}

In general, a potential theoretical weakness of the FCS method is that the implicit joint distribution underlying the separate models may not always exist, i.e., the conditional models may be incompatible. However, in our proposed FCS approach we note that the conditional logistic regression models are compatible with, and implied by, the loglinear joint distribution. As a result, imputations created in Algorithm \ref{alg:fcs-main} are draws from the implicit loglinear joint distribution. \cite{liu2014stationary} showed that imputation created by compatible conditional models is asymptotically equivalent to full Bayesian imputation for the assumed joint model. \cite{hughes2014joint} note that, in finite samples, an additional condition is required. They prove that compatible conditional models, together with a ``non-informative margins'' condition, is sufficient for the FCS algorithm to impute missing data from the predictive distribution of the missing data implied by the joint model. When the ``non-informative margins'' condition is not satisfied, order effects (i.e., systematic differences depending upon the ordering of variables in the FCS algorithm) can occur, although results of simulations indicate that 
the average magnitude of the order effects is small and does not induce bias \citep{hughes2014joint}. The efficiency, on the other hand, may suffer a modest reduction \citep{seaman2018relative}.

Finally, we consider the case where there are always-observed covariates of interest $X$. The goal is to relate $Y_1,\ldots, Y_K$ to $X$, and the FCS-NSC algorithm operates under a minor modification of the NSC assumption. Specifically, the conditional independence assumption also involves conditioning on $X$:
$$
Y_k\indep M_k | Y_{-k}, M_{-k}, X.
$$
The loglinear representation in (\ref{eq:NSC-generic}) still holds for the conditional distribution $(Y,M)|X$, where the loglinear model parameters $\lambda$'s are now functions of $X$, denoted by $\lambda(X)$:
$$
p(M,Y|X) \propto \exp\left(\sum_{I\in \mathcal P([K])}\sum_{J\in \mathcal P([K]\setminus I)}\lambda_{M_IY_J}(X)M_IY_J\right).
$$
As a result, the regression parameters in the logistic model (\ref{eq:logistic-NSC}) also depends on $X$. When $X$ is continuous, we propose to incorporate $X$ in Algorithm \ref{alg:fcs-main} by adding its main effect and interactions with $Y_{-k}$ and $M_{-k}$ to the logistic models in line 13 and impute $Y_k$ by sampling from the expanded logistic models in line 14. This approach is equivalent to assuming that $\lambda(X)$ is a linear function of $X$. When $X$ is discrete, we can stratify the dataset with respect to $X$ and impute within each stratum using FCS-NSC.

\subsection{Sensitivity Analysis under NSC}
Next, we describe how to conduct sensitivity analysis anchored at NSC. We argue that NSC is a better anchoring assumption compared to MAR for non-monotone and MNAR data because: 1) it is exceedingly challenging to conceive of models for non-monotone missingness that do not impose more restrictive assumptions than what MAR strictly entails; the few existing MAR models for non-monotone missingness involve complex specifications for the conditional probabilities of different missing patterns \citep{robins1997non,sun2018inverse}; 2) these MAR models typically do not have a set of clinically interpretable parameters through which departures from MAR can be achieved. On the contrary, NSC leads to simple forms of the fully conditional distributions (see for example Eq. \ref{eq:logistic-NSC}) and a natural set of sensitivity parameters that are readily interpretable by study investigators while preserving the computational efficiency of the imputation approach.

Although the NSC assumption allows the probability that a particular variable is missing to depend on the value of other unobserved variables, in many applications one might want to consider MNAR mechanisms where the probability that a particular variable is missing depends on the (possibly unobserved) value of that variable. The latter dependence can be expressed in terms of a set of sensitivity parameters. Recall the connection between NSC and the loglinear model in (\ref{eq:loglinear-NSC}). By adding interaction terms to the loglinear model that are omitted by the NSC assumption, we can produce a model for $(M,Y)$ that deviates arbitrarily from NSC. The number of such interaction terms is large when $K>5$ and we focus on the set of interaction terms of the lowest order: $\lambda_{M_kY_k}M_kY_k$ for $k = 1,\ldots, K$. We note that these parameters for the conditional association between $M_k$ and $Y_k$ are not identifiable based on the observed data and therefore are natural candidates for sensitivity parameters. With these additional terms, equation (\ref{eq:logistic-NSC}) now becomes
$$
\begin{aligned}
    \text{logit}(p(Y_k=1|Y_{-k}, M)) = &\lambda_{Y_k} + \lambda_{M_kY_k}M_k + \sum_{l\neq k} (\lambda_{M_lY_k}M_l+\lambda_{Y_lY_k}Y_l)\\
    +& \sum_{l\neq r\neq k}\lambda_{Y_k Y_lM_r} Y_lM_r + \lambda_{Y_kY_lY_r}Y_lY_r + \lambda_{Y_kM_lM_r}M_lM_r.
\end{aligned}
$$
The sensitivity parameter $\lambda_{M_kY_k}$, assumed to be fixed and known, indicates the conditional log odds ratio of $Y_k$ being missing when $Y_k=1$ versus when $Y_k=0$. For example, if $Y_k$ represents the occurrence of certain clinical events, such as abstinence from a drug at time $k$, this parameter has a clear clinical interpretation -- the difference in log odds of not responding at time $k$ between abstinent and and non-abstinent individuals, and thus its plausible range of values can be inferred from domain knowledge.

For simplicity, we introduce the sensitivity analysis based on Algorithm \ref{alg:fcs-main} by assuming that $\lambda_{M_kY_k} = \lambda_{MY}$ for $k=1,\ldots,K$. Note that since the augmented logistic regression in line 13 of Algorithm \ref{alg:fcs-main} is applied to subjects with $M_k = 0$, a non-zero sensitivity parameter $\lambda_{MY}$ does not change the results for this step. However, when predicting the missing values of $Y_k$, we need to use the sensitivity parameter $\lambda_{MY}$:
$$
\text{logit}(\hat p_\lambda(Y_k|Y_{-k}, M_{-k})) = \text{logit}(\hat p(Y_k|Y_{-k}, M_{-k})) + \lambda_{MY}.
$$

\subsection{Relationship between NSC and MAR}
In this section, we discuss in detail the connection, and the distinction between NSC and MAR assumptions, given that MAR is the most commonly assumed missing data mechanism for conventional FCS. We argue that when the true missing data mechanism is MAR, imputation using the proposed FCS-NSC algorithm (Algorithm \ref{alg:fcs-main}) will introduce bias and {\em vice versa}. Even though it appears that Algorithm \ref{alg:fcs-main} is no more complex than the conventional FCS-MAR algorithm, except for the inclusion of additional missingness indicators $M_{-k}$ in the fully conditional models, we show that these redundant and seemingly harmless terms will lead to incorrect sampling distributions for imputation.

Recall that a necessary condition for the validity of any FCS algorithm is specified through the equalities
\begin{equation}
p(Y_k|\cdot, M_k = 1) = p(Y_k|\cdot, M_k = 0), \forall k = 1,\ldots,K,
\label{eq:FCS-validity}
\end{equation}
where $\cdot$ indicates the set of covariates in the FCS model. In FCS-MAR, $Y_{-k}$ are the corresponding covariates; whereas in FCS-NSC, $Y_{-k}$ and $M_{-k}$ are the corresponding covariates. Equalities (\ref{eq:FCS-validity}) ensure that the fully conditional models trained with samples where $Y_k$ is observed approximate well the fully conditional distributions of $Y_k$ in samples where $Y_k$ is missing. As a result, the imputation step of an FCS algorithm samples from the correct conditional distribution. To show the validity of the FCS-MAR algorithm under MAR, note that the definition of MAR implies that for any $k\in[K]$, 
$
p(M_1,\ldots,M_k = 1,M_{k+1},\ldots,M_K|Y_1,\ldots, Y_K)
$ does not depend on $Y_k$ for any values of $M_{-k}$. Hence  $\sum_{M_{-k}}p(M_{-k},M_k = 1|Y_1,\ldots, Y_K)$ also does not depend on $Y_k$. In other words, $p(M_k = 1|Y_1,\ldots,Y_K) = p(M_k=1|Y_{-k})$. This means that $M_k$ and $Y_k$ are conditionally independent given $Y_{-k}$ for $k=1,\ldots,K$ and thus (\ref{eq:FCS-validity}) holds. This set of equalities for FCS-NSC under NSC can be easily verified using the assumption in (\ref{eq:NSC-assumption}).

If we perform imputation with FCS-MAR when the true mechanism is NSC, these equalities generally do not hold and the imputation results are biased. However, with carefully chosen values of the $\lambda$'s, an NSC model can satisfy them without reducing to MCAR. We prove in Section S3 of the Supplementary Materials that even under these special cases, imputing data with FCS-MAR when the true 
mechanism is NSC still leads to biased results. Next, we focus on showing that under MAR but not MCAR, equality (\ref{eq:FCS-validity}) with $Y_{-k}$ and $M_{-k}$ as covariates does not hold, which means using FCS-NSC for MAR data will generate biased results. This conclusion is a direct consequence of the following theorem.

\begin{theorem}
All NSC models for binary outcomes $(Y_1,\ldots,Y_K)$ that also satisfy the MAR assumption are MCAR.
\label{thm:NSC-MAR}
\end{theorem}

\begin{proof}
We start by writing out the conditional distribution $p(M|Y)$ under the NSC assumption with the loglinear representation:
\begin{equation}
\begin{aligned}
p(M|Y) &= \frac{p(M, Y)}{p(Y)} = \frac{p(M, Y)}{\sum_{m \in \mathcal M_K} p(M = m, Y)}\\
& = \frac{\exp\left(\sum_{I\in \mathcal P([K])\setminus[K] }\sum_{J\in \mathcal P^+([K]\setminus I)}\lambda_{Y_IM_J}Y_IM_J\right)}{\sum_{m\in \mathcal M_K}
\exp\left(\sum_{I\in \mathcal P([K])\setminus[K] }\sum_{J\in \mathcal P^+([K]\setminus I)}\lambda_{Y_IM_J}Y_Im_J\right)},
\end{aligned}
\label{eq:NSC-MAR-proof}
\end{equation}
where $\mathcal M_K$ denotes the set of all possible values of a $K-$dimensional binary vector. Since $Y$ is conditioned upon, all terms involving only $Y$'s are cancelled out. The remaining terms are either $M$-$Y$ interactions or $M$-$M$ interactions (including main effects of $M$).

Denote $e_k$ as the $K$-dimensional vector with the $k$-th element equal to one and all other elements equal to zero. If the model is also MAR, $P(M=m|Y)$ only depends on $Y_{(m)}$, the observed sub-vector of $Y$. Hence $P(M = e_k|Y)$ does not depend on $Y_k$. In this case, the numerator of (\ref{eq:NSC-MAR-proof}) is free of $Y_k$, since by the NSC assumption, the only $M$-$Y$ interaction terms with $Y_k$ involved are those including exclusively $M_{k'}, k\neq k'$. But $M_{k'} = 0$ for all $k'\neq k$. Therefore, if an NSC model is also MAR, the denominator of (\ref{eq:NSC-MAR-proof}) does not depend on $Y_k$. This holds for all $k = 1,\ldots, K$ and thus the denominator of (\ref{eq:NSC-MAR-proof}) is a constant. Note that the denominator being a constant does not necessarily mean that all $\lambda$'s for $M$-$Y$ interactions are zero. Indeed, the number of $\lambda$'s corresponding to $M$-$Y$ interactions is $3^K - 2^{K+1} + 1$. It is larger than the number of equations implied by the condition of constant denominator, which is $2^K - 1$, when $K>2$. Hence even though $\lambda_{Y_IM_J}=0$ for all $M-Y$ interactions is a solution to the equations, it is not unique.

Since the denominator is a constant, when $m = e_k + e_{k'}$, $k\neq k'$, the MAR assumption holds if and only if the numerator of (\ref{eq:NSC-MAR-proof}) does not depend on either $Y_k$ or $Y_{k'}$. We can write out explicitly the numerator:
\begin{align*}
f(Y;\lambda) = &\exp\left(\sum_{I\in \mathcal P([K])\setminus[K] }\sum_{J\in \mathcal P^+([K]\setminus I)}\lambda_{Y_IM_J}Y_Im_J\right)\\
=& \exp\left(\sum_{I\in \mathcal P([K]\setminus \{k\})}\lambda_{Y_IM_k}Y_I + \sum_{J\in \mathcal P([K]\setminus \{k'\})}\lambda_{Y_JM_{k'}}Y_J + \sum_{L\in \mathcal P([K]\setminus \{k,k'\})} \lambda_{Y_L M_k M_{k'}}Y_L \right).
\end{align*}
The MAR assumption implies that
\begin{equation}
f(Y_k = 0, Y_{k'}=0, Y_{-\{k,k'\}}=0;\lambda) = f(Y_k = 1, Y_{k'}=0, Y_{-\{k,k'\}}=0;\lambda).
\label{eq:NSC-MAR-proof-pair}
\end{equation}
It follows that
$
\lambda_{Y_{k'}M_k} = 0.
$
Similarly, we have $\lambda_{Y_kM_{k'}}=0$. Repeatedly applying this approach, we can prove that $\lambda_{Y_kM_{k'}}=0$ for all $(k,k')$ pairs with $k\neq k'$. Now if we set $Y_{k''}=1$ for an arbitrary $k''\neq k$ and $k''\neq k'$, we can show that $\lambda_{Y_{k'}Y_{k''}M_k}=0$ for all triplet $(k, k', k'')$ with $k\neq k'\neq k''$. Sequentially flipping one $Y_l$ from zero to one in (\ref{eq:NSC-MAR-proof-pair}), we can prove that $\lambda_{Y_IM_k}=0$ for all $I\in \mathcal P^+([K]\setminus k)$ and $k = 1,\ldots,K$.

When applying the above process to an $m$ with $K_0\geq 3$ elements equal to 1, we can follow the same procedure above and show that the $\lambda$'s with $K_0$ $M$'s in the subscripts are all zero. In other words, $\lambda_{Y_IM_J}= 0$ where $J\subset[K]$ with cardinality $K_0$ and $I\in \mathcal P^+([K]\setminus J)$. Thus, $\lambda_{Y_IM_J} = 0$, for all $I\in\mathcal P^+([K])\setminus[K]$ and $J\in \mathcal P^+([K]\setminus I)$. This means that if an NSC model is also MAR, then its loglinear representation only includes main effects of $M$ and $Y$ as well as $M$-$M$ and $Y$-$Y$ interactions. Such a model is indeed MCAR as $p(M|Y)$ does not depend on $Y$.
\end{proof}

\begin{corollary}
For any MAR model for binary outcomes $(Y_1,...,Y_K)$ that is not MCAR, there exists at least one $\lambda_{Y_IM_J}\neq 0$, where $I, J\subset[K]\setminus\emptyset$ and $I\cap J\neq \emptyset$.
\label{cor:MAR-NSC}
\end{corollary}

\begin{proof}
We first note that any joint distribution of $K$ binary random variables and their missingness indicators under a MAR assumption can be represented by a saturated loglinear model. If an MAR model is not MCAR, then it cannot satisfy the NSC assumption by Theorem \ref{thm:NSC-MAR}. The proof is completed by the fact that the NSC assumption is equivalent to the restrictions on a loglinear model that $\lambda_{Y_IM_J}=0$ for any $I\cap J\neq \emptyset$.
\end{proof}

Without loss of generality, assuming for a MAR model (that is not also MCAR) the non-zero $\lambda_{Y_IM_J}$ with $I\cap J\neq \emptyset$ also satisfies $k_0\in I\cap J$, it follows from (\ref{eq:NSC-generic}) that $\text{logit}(p(Y_{k_0}|Y_{-k_0},M))$ is also a function of $M_{k_0}$ and thus
$$
p(Y_{k_0}|Y_{-k_0}, M_{-k_0}, M_{k_0}=1) \neq p(Y_{k_0}|Y_{-k_0}, M_{-k_0}, M_{k_0}=0).
$$
The FCS algorithm under NSC applies an incorrect sampling distribution for $Y_{k_0}$ if the missing mechanism is MAR (but not also MCAR) and will generate biased imputation results. From the perspective of applications of FCS models, this result suggests that including missingness indicators in FCS models should not be considered as a default or an automatic strategy for imputation of missing data in every setting.

Finally, in proposing NSC as an alternative to MAR for the analysis of data that are missing non-monotonically, it is worth emphasizing some important distinctions between NSC and MAR.
Recall that MAR is fundamentally defined in terms of an assumption about the probabilities of particular missingness patterns, $p(M | Y)$ but not 
 the probabilities that an individual variable is missing, $p(M_k |Y)$. Specifically, MAR requires that the probability of the {\em observed} missingness pattern does not depend on any unobserved values of $Y$, i.e., $p(M | Y)=p(M | Y_{(M)})$. As a result, one can estimate any parameters related to the joint distribution of $Y$ based solely on $p(Y_{M})$ \citep{little2019statistical}, which can be done using the conventional FCS-MAR. We note that the MAR model for $p(M | Y_{(M)})$ given in \cite{lin2018exact} (and used in our simulations below) are contrived linear regression models with unusual constraints to ensure these joint conditional MAR probabilities for $p(M | Y_{(M)})$ are between 0 and 1.  Note, the joint model for $(M, Y)$ under MAR (but not MCAR) cannot be a non-saturated loglinear model \citep{kimandkim2017}. If one tries to coerce or force the joint model for $(M, Y)$ under MAR to be a loglinear model under NSC, then we actually force some of the loglinear parameters corresponding to interaction(s) between $Y$ and $M$ to
be non-zero, i.e., force the MAR model (which cannot be a non-saturated loglinear model) to be an MNAR loglinear model.

In contrast to MAR, NSC can be defined in terms of an assumption about 
$p(M_{k} | M_{-k}, Y)$; specifically it requires the conditional independence of $M_k$ and $Y_k$, but allows dependence on unobserved values of all other outcomes (and all other missingness indicators) and thus is a MNAR mechanism.  As was highlighted earlier,  MAR is emphatically not a special case of NSC; put another way, for a NSC model to be MAR, it must also be MCAR (with $p(M | Y)$ not depending on $Y$). Fundamentally, NSC and MAR are distinct assumptions about the missing data mechanism; consequently, and as will be illustrated in finite sample and asymptotic studies in Section 3, a misspecified FCS algorithm based on assuming MAR will yield biased estimates when the true model is NSC and {\em vice versa}.

\section{Simulation and Asymptotic Studies}
In this section, we illustrate the performance and properties of the FCS-NSC algorithm in three distinct settings. For the first two settings, we use loglinear models to generate $(M,Y)$ under a NSC missing mechanism, while in the third, we simulate missing data under a MAR assumption. We set $K=6$ in all three scenarios to reflect the dimension of the outcome variables in our data application in Section 4; we describe the details of the three simulation scenarios below. In addition, we supplement the simulations with studies of asymptotic bias in the same scenarios. The R code to replicate all analyses is available in the Supplementary Materials.

\begin{enumerate}
    \item \textbf{Main effect NSC}. We specify the joint distribution of $(M,Y)$ based on the generic loglinear representation in (\ref{eq:NSC-generic}).
    We place the constraints that $\lambda_{Y_IM_J} = 0$ when $|I|>1$ and $|J|>1$ or $|I| + |J|>2$. This implies that the logistic regression model of $Y_k|Y_{-k},M_{-k}$ only contains main effects of $Y_{-k}$ and $M_{-k}$ (see Eq. \ref{eq:logistic-NSC}). We fix $\lambda_{Y_kY_{k'}}=0.5$ for $k\neq k' = 1,\ldots,K$ and let $\lambda_{Y_kM_{k'}}=2$ for $k' \in \{1, 2, 3\}$ and $\lambda_{Y_kM_{k'}}=-2$ for $k'\in \{4, 5, 6\}$. We specify the values of the other non-zero $\lambda$ parameters by solving a system of equations such that $P(Y_k = 1) = 0.4$ when $k\leq 3$ and $P(Y_k = 1) = 0.6$ when $k>3$. We set the marginal missingness rate, $P(M_k)$, to be a fixed values across all $k$ and three different values $0.2, 0.3, 0.4$ are considered.

    \item \textbf{$Y$-$M$ interaction NSC}. In this scenario, we 
    expand the previous scenario to allow for additional non-zero $\lambda$ parameters in (\ref{eq:NSC-generic}). Specifically, we fix all $\lambda_{Y_kY_{k'}M_{k''}}$ to be a non-zero constant. This indicates that the logistic model for $Y_k|Y_{-k},M_{-k}$ now involves not only the main effects of $Y_{-k}$ and $M_{-k}$ but also interaction terms $Y_kM_{k'}$. We consider three possible values ($-0.5,-1,-2$) for these order-3 parameters. The $\lambda_{Y_kY_{k'}}$ parameters are still set to be 0.5 for all $k\neq k'=1,\ldots,K$ and we choose the other non-zero $\lambda$ parameters such that the marginal probabilities $P(Y_k = 1)$ are 0.4 when $k \leq 3$ and $P(Y_k = 1) = 0.6$ when $k > 3$. We fix the marginal missingness rate at 0.3 for all outcome variables.

    \item \textbf{MAR}. As we have mentioned in Section 2, it is not straightforward to specify a non-monotone MAR model with $K = 6$. Therefore, we construct a non-monotone MAR model by assuming that $(Y_k,Y_{k+3},M_k,M_{k+3})$, $k = 1,2,3$ are independent of each other and each group follows the MAR model in \cite{lin2018exact}. In particular, we assume that the conditional distribution of $M_{k}, M_{k+3}$ given $Y_k,Y_{k+3}$ is
    \begin{equation}
        \begin{aligned}
        Pr(M_k=1,M_{k+3}=0|Y_k,Y_{k+3}) &= w_1Y_{k+3} + w_2(1-Y_{k+3}),\\
        Pr(M_k=0,M_{k+3}=1|Y_k,Y_{k+3}) &= v_1Y_k + v_2(1-Y_k),\\
        Pr(M_k=0,M_{k+3}=0|Y_k,Y_{k+3}) &= \pi_0.
        \end{aligned}
        \label{eq:MAR-bivar}
    \end{equation}
    For the joint distribution of $Y_1,\ldots,Y_K$,  we use a loglinear model with second-order interactions among the $Y_k$'s and set $\lambda_{Y_kY_{k'}} = 0.5$ or $2$.  Thus, the joint distribution of $M$ and $Y$ under MAR is specified by the conditional distribution of $M$ given $Y$ and the marginal distribution of $Y.$ We select the values of all other parameters such that the same marginal probabilities of $Y_k$, as well the marginal missingness rates, as in the first simulation scenario are obtained. We note that we used a set of three bivariate MAR mechanisms because more general non-monotone MAR mechanisms require many constraints on $f(M|Y)$ and, in general, do not model realistic beliefs about the missing data process in most substantive applications \citep{robins1997non}. Although this MAR mechanism may be somewhat contrived and unrealistic, it nonetheless allows us to examine the robustness of FCS-NSC when the true missingness mechanism is MAR.  
    
\end{enumerate}

In the simulation and asymptotic studies, we consider three approaches for handling the missing data: 1) FCS-NSC algorithm, 2) FCS-MAR algorithm and 3) available-case analysis (using all observed data at a given time point for estimation). We estimate the fully conditional distribution of $Y_k$ in the FCS-NSC algorithm via logistic regression and include only main effects of $Y_{-k}$ and $M_{-k}$ in the model. The FCS-MAR algorithm operates similarly, except we use only main effects of $Y_{-k}$ in the logistic model of $Y_k|Y_{-k}$. We benchmark the three approaches based on their accuracy in estimating the marginal probabilities of $Y_k$, $k = 1,\ldots,6$. Specifically, we examine the performance under finite sample size ($N = 200$) that is consistent with the sample size in our data application in Section 4, and also asymptotically, where $N\to\infty$ and imputation is directly applied to the underlying distribution of $(M,Y_{(M)})$. The details of the algorithm used to derive the asymptotic results are deferred to Section S2 in the Supplementary Materials. We also have explored nonparametric models such as random forest and support vector machine as alternatives for logistic regression for the two FCS algorithms, with the hope that they would better address the issue of model misspecification (e.g., detection of higher order interactions as in the $Y$-$M$ interaction NSC scenario). However, the results are very similar to those based on logistic regression and thus are not reported in the manuscript.

\subsection{Results for Main Effect NSC}
We summarize the performance of the three approaches based on the relative bias of their respective estimates (see Table \ref{tab:nsc-correct}). In this scenario, the FCS-NSC algorithm uses correctly specified models for the conditional distributions $Y_k|Y_{-k},M_{-k}$. The results confirm that the FCS-NSC has the lowest relative bias for all outcomes and all missing rates with finite sample size, and is also asymptotically unbiased. Concordant with the discussion in Section 2.4, FCS-MAR leads to biased estimates when the true missingness mechanism is NSC. The bias is relatively small when the missing rate is low (0.2) but increases dramatically when the missing rate reaches 0.4. Its asymptotic behavior has the same general pattern and the relative biases are not too different from those under finite sample size. The available-case analysis yields the highest biases across all missing rates, both for finite sample size and asymptotically.

\begin{table}[htbp]
  \centering
  \caption{Percent bias for the marginal probabilities of $Y_{k}$ associated with the three different methods when the true missingness mechanism is NSC (and FCS-NSC is correctly specified). Results are presented for both finite sample size ($N = 200$) and asymptotically ($N\to\infty$), where finite sample results are averages over 1,000 simulation replicates. Note: True marginal probabilities for Outcomes 1$-$3 are 0.4 and for Outcomes 4$-$6 are 0.6.\\}
  \begin{adjustbox}{max width=\textwidth}
    \begin{tabular}{|c|c|c|c|c|c|c|c|c|c|c|}
    \hline
    \multicolumn{2}{|c|}{\multirow{2}[4]{*}{Scenario}} & \multicolumn{3}{c|}{missing rate = 0.2} & \multicolumn{3}{c|}{ missing rate = 0.3} & \multicolumn{3}{c|}{missing rate = 0.4} \bigstrut\\
\cline{3-11}    \multicolumn{2}{|c|}{} & MAR   & NSC   & Available & MAR   & NSC   & Available  & MAR   & NSC   & Available  \bigstrut\\
    \hline
    \multirow{6}[2]{*}{N = 200} & Outcome 1 & -2.51 & -0.10 & -26.38 & -8.23 & 0.70  & -47.58 & -31.61 & 1.28  & -70.99 \bigstrut[t]\\
          & Outcome 2 & -2.72 & -0.23 & -26.59 & -7.96 & 0.77  & -47.07 & -32.10 & 0.05  & -71.50 \\
          & Outcome 3 & -2.58 & -0.15 & -26.28 & -7.86 & 0.88  & -46.94 & -30.99 & 0.84  & -71.20 \\
          & Outcome 4 & 1.87  & 0.26  & 17.75 & 5.52  & -0.46 & 31.54 & 21.20 & -0.97 & 47.52 \\
          & Outcome 5 & 1.90  & 0.28  & 17.92 & 5.54  & -0.44 & 31.64 & 20.94 & -0.60 & 47.53 \\
          & Outcome 6 & 1.81  & 0.28  & 17.54 & 5.44  & -0.18 & 31.53 & 20.95 & -0.10 & 47.57 \bigstrut[b]\\
    \hline
    \multirow{2}[2]{*}{Asymptotic} & Outcome 1-3 & -2.67 & 0.00  & -26.47 & -8.92 & 0.00  & -47.30 & -36.32 & 0.00  & -71.37 \bigstrut[t]\\
          & Outcome 4-6 & 1.78  & 0.00  & 17.65 & 5.94  & 0.00  & 31.53 & 24.18 & 0.00  & 47.57 \bigstrut[b]\\
    \hline
    \end{tabular}%
    \end{adjustbox}
  \label{tab:nsc-correct}%
\end{table}%

\subsection{Results for $Y$-$M$ Interaction NSC}
Note that the additional three-way interactions introduced in this scenario make the logistic models assumed in the FCS-NSC algorithm misspecified. As a result, the estimates from FCS-NSC are biased both under finite sample size and asymptotically (see Table \ref{tab:nsc-mis}). When the sample size is finite, the bias increases with the strength of the three-way interactions $|\lambda_{Y_kY_{k'}M_{k''}}|$. However, we observe a much larger discrepancy between the finite sample bias and asymptotic bias for FCS-NSC in this scenario compared to the main effect NSC scenario: the asymptotic bias for Outcome 1$-$3 remains at a low level with no apparent trend in its magnitude as $|\lambda_{Y_kY_{k'}M_{k''}}|$ increases while the asymptotic bias for Outcome 4$-$6 is non-monotone as $|\lambda_{Y_kY_{k'}M_{k''}}|$ increases.  The bias of FCS-NSC is the smallest among all three approaches for most configurations when $N = 200$ and universally outperforms the other two approaches asymptotically. This suggests that when the missing data mechanism is NSC, FCS-NSC with some degree of misspecification of the fully conditional models may be preferred over imputation approaches that rely on MAR or MCAR assumptions. 

As expected, FCS-MAR gives biased results and the magnitude of the bias is not a monotone function of $|\lambda_{Y_kY_{k'}M_{k''}}|$. Indeed, both the finite sample bias and the asymptotic bias first decrease and then increase in magnitude as $|\lambda_{Y_kY_{k'}M_{k''}}|$ gets larger. We also note the magnitude of the discrepancy between the 
finite sample bias and asymptotic bias is smaller for FCS-MAR than for FCS-NSC. Overall, the magnitude of the bias of FCS-MAR is larger than FCS-NSC but discernibly smaller than for the available-case estimator. The available-case analysis has the highest bias both under finite sample size and asymptotically; the bias also follows the ``decreasing-increasing'' trend as for FCS-MAR. This non-monotonicity is common in data generated with loglinear model where the higher order interactions have opposite signs as the lower order interactions \citep{fitzmaurice1993regression}.

\begin{table}[htbp]
  \centering
  \caption{Percent bias for the marginal probabilities of $Y_{k}$ associated with the three different methods when the true missingness mechanism is NSC (and FCS-NSC is misspecified). Results are presented for both finite sample size ($N = 200$) and asymptotically ($N\to\infty$), where finite sample results are averages over 1,000 simulation replicates. Note: True marginal probabilities for Outcomes 1$-$3 are 0.4 and for Outcomes 4$-$6 are 0.6.\\}
  \begin{adjustbox}{max width=\textwidth}
    \begin{tabular}{|c|c|c|c|c|c|c|c|c|c|c|}
    \hline
    \multicolumn{2}{|c|}{\multirow{2}[4]{*}{Scenario}} & \multicolumn{3}{c|}{$\lambda_{L_iL_jR_k}= -0.5$} & \multicolumn{3}{c|}{$\lambda_{L_iL_jR_k}= -1$} & \multicolumn{3}{c|}{$\lambda_{L_iL_jR_k}= -2$} \bigstrut\\
\cline{3-11}    \multicolumn{2}{|c|}{} & MAR   & NSC   & Available  & MAR   & NSC   & Available  & MAR   & NSC   & Available  \bigstrut\\
    \hline
    \multirow{6}[2]{*}{$N = 200$} & Outcome 1 & -11.47 & 2.69  & -44.03 & -2.38 & 5.92  & 6.20  & 14.60 & 13.16 & 55.62 \bigstrut[t]\\
          & Outcome 2 & -11.19 & 2.77  & -43.88 & -2.92 & 5.31  & 5.53  & 14.60 & 13.39 & 55.74 \\
          & Outcome 3 & -10.74 & 3.21  & -43.63 & -2.50 & 5.96  & 6.22  & 14.51 & 13.01 & 55.20 \\
          & Outcome 4 & 25.73 & -1.86 & 51.19 & 9.27  & -2.27 & 25.93 & -8.63 & -11.51 & 19.58 \\
          & Outcome 5 & 25.22 & -1.89 & 51.17 & 9.84  & -2.10 & 26.16 & -8.62 & -11.83 & 19.91 \\
          & Outcome 6 & 24.94 & -2.05 & 51.15 & 9.84  & -2.29 & 26.38 & -8.70 & -11.55 & 19.81 \bigstrut[b]\\
    \hline
    \multirow{2}[2]{*}{Asymptotic} & Outcome 1-3 & -14.31 & 0.17  & -43.97 & -5.80 & 0.18  & 5.82  & 10.95 & -0.19 & 55.92 \bigstrut[t]\\
          & Outcome 4-6 & 33.37 & 1.15  & 51.27 & 11.34 & 6.02  & 26.36 & -8.98 & -14.22 & 20.11 \bigstrut[b]\\
    \hline
    \end{tabular}%
    \end{adjustbox}
  \label{tab:nsc-mis}%
\end{table}%

\subsection{Results for MAR}
We summarize the results for this simulation scenario in Table \ref{tab:mar}. The results for FCS-NSC confirm the statement made in Section 2.4 that imputing MAR data with FCS-NSC will inflict bias on the estimates. Interestingly, the magnitude of the bias tends to be smaller when the pairwise dependence between $Y_k$ and $Y_{k'}$ increases; however, this result cannot be expected in general. One possible explanation for this phenomenon is that with higher $\lambda_{Y_kY_{k'}}$, the misspecified logistic models in FCS-NSC can capture more accurately the actual conditional distribution $Y_k|Y_{-k},M$ as the dependency between $Y_k$ and $M_k$, which is non-zero based on Theorem \ref{thm:NSC-MAR} for MAR data, can be absorbed into the main effects of $Y_{-k}$ and $M_{-k}$ through the dependence between the $Y$'s. As might be expected, in this scenario the FCS-MAR algorithm  produces the least biased estimates under finite sample size and is 
asymptotically unbiased. Similar to the previous two scenarios, the available-case analysis yields the most biased results and should, in general, be avoided in real data applications.

\begin{table}[htbp]
  \centering
  \caption{Percent bias for the marginal probabilities of $Y_{k}$ associated with the three different methods when the true missingness mechanism is MAR (and FCS-NSC is misspecified). Results are presented for both finite sample size ($N = 200$) and asymptotically ($N\to\infty$), where finite sample results are averages over 1,000 simulation replicates. Note: True marginal probabilities for Outcomes 1$-$3 are 0.4 and for Outcomes 4$-$6 are 0.6.\\}
  \begin{adjustbox}{max width=\textwidth}
    \begin{tabular}{|c|c|c|c|c|c|c|c|c|c|c|}
    \hline
    \multicolumn{11}{|c|}{$\lambda_{L_iL_j} = 0.5$} \bigstrut\\
    \hline
    \multicolumn{2}{|c|}{\multirow{2}[4]{*}{Scenario}} & \multicolumn{3}{c|}{missing = 0.2} & \multicolumn{3}{c|}{missing = 0.3} & \multicolumn{3}{c|}{missing = 0.4} \bigstrut\\
\cline{3-11}    \multicolumn{2}{|c|}{} & MAR   & NSC   & Available & MAR   & NSC   & Available & MAR   & NSC   & Available \bigstrut\\
    \hline
    \multirow{6}[2]{*}{N = 200} & Outcome 1 & 0.45  & -0.64 & 1.49  & 0.33  & -7.05 & 3.13  & 1.94  & -3.21 & 6.05 \bigstrut[t]\\
          & Outcome 2 & 0.79  & -0.30 & 1.77  & 0.60  & -6.89 & 3.29  & 2.10  & -3.15 & 6.18 \\
          & Outcome 3 & 0.80  & -0.24 & 1.95  & 1.37  & -6.17 & 3.93  & 1.96  & -3.06 & 6.23 \\
          & Outcome 4 & 0.01  & 0.70  & -0.65 & -0.61 & 4.33  & -2.56 & -1.23 & 2.05  & -3.91 \\
          & Outcome 5 & -0.25 & 0.31  & -0.99 & -1.19 & 3.91  & -3.08 & -1.73 & 1.95  & -4.38 \\
          & Outcome 6 & -0.11 & 0.56  & -0.83 & -0.60 & 4.33  & -2.54 & -1.60 & 1.92  & -4.29 \bigstrut[b]\\
    \hline
    \multirow{2}[2]{*}{Asymptotic} & Outcome 1-3 & 0.00  & -1.42 & 1.45  & 0.00  & -9.18 & 3.74  & 0.00  & -11.87 & 6.06 \bigstrut[t]\\
          & Outcome 4-6 & 0.00  & 0.94  & -0.97 & 0.00  & 6.12  & -2.49 & 0.00  & 7.94  & -4.04 \bigstrut[b]\\
    \hline
    \multicolumn{11}{|c|}{$\lambda_{L_iL_j}=2$} \bigstrut\\
    \hline
    \multicolumn{2}{|c|}{\multirow{2}[4]{*}{Scenario}} & \multicolumn{3}{c|}{missing = 0.2} & \multicolumn{3}{c|}{missing = 0.3} & \multicolumn{3}{c|}{missing = 0.4} \bigstrut\\
\cline{3-11}    \multicolumn{2}{|c|}{} & MAR   & NSC   & Available & MAR   & NSC   & Available & MAR   & NSC   & Available \bigstrut\\
    \hline
    \multirow{6}[2]{*}{N = 200} & Outcome 1 & 0.66  & 0.59  & 4.44  & 1.89  & 0.84  & 12.29 & 4.35  & 7.67  & 20.72 \bigstrut[t]\\
          & Outcome 2 & 0.77  & 0.73  & 4.67  & 2.07  & 1.02  & 12.56 & 4.69  & 8.06  & 21.21 \\
          & Outcome 3 & 0.80  & 0.69  & 4.69  & 1.76  & 0.82  & 12.39 & 4.62  & 7.64  & 21.14 \\
          & Outcome 4 & -0.98 & -1.00 & -3.49 & -1.69 & -1.06 & -8.50 & -2.71 & -5.03 & -13.78 \\
          & Outcome 5 & -0.88 & -0.79 & -3.26 & -1.68 & -1.06 & -8.66 & -3.00 & -5.24 & -13.98 \\
          & Outcome 6 & -1.00 & -0.94 & -3.67 & -1.69 & -1.04 & -8.72 & -2.84 & -4.95 & -13.70 \bigstrut[b]\\
    \hline
    \multirow{2}[2]{*}{Asymptotic} & Outcome 1-3 & 0.00  & -0.31 & 4.96  & 0.00  & -1.56 & 12.75 & 0.00  & 1.21  & 20.66 \bigstrut[t]\\
          & Outcome 4-6 & 0.00  & 0.21  & -3.31 & 0.00  & 1.04  & -8.50 & 0.00  & -0.73 & -13.78 \bigstrut[b]\\
    \hline
    \end{tabular}%
    \end{adjustbox}
  \label{tab:mar}%
\end{table}%

\section{Application}

We illustrate the use of the proposed FCS-NSC algorithm with an application to data collected from a longitudinal clinical trial comparing different treatments for cocaine use disorder, the National Institute on Drug Abuse (NIDA) Collaborative Cocaine Treatment Study (CCTS) \citep{crits1999psychosocial}. Briefly, the NIDA CCTS was a multi-site clinical trial of $N = 487$ patients randomized to four psychosocial treatments for six months: group drug counseling (GDC) alone, individual cognitive therapy (CT) plus GDC, individual supportive-expressive (SE) psychodynamic therapy plus GDC, and individual drug counseling (IDC) plus GDC. In this trial, participants were 18 years of age or older (mean age 33.9 yrs), 23\% female, 58\% white, and met criteria for current cocaine dependence according to Diagnostic and Statistical Manual of Mental Disorders, Fourth Edition \citep{association2000diagnostic}. The trial utilized both weekly urine toxicology screens and self-report information at the monthly research visits to determine the status of cocaine use (yes/no) in each of the six months of active treatment. We define the endpoint of our analysis to be abstinence {\em for at least three consecutive months} during the six-month treatment, which is consistent with a criterion used in the original study \citep{crits1999psychosocial}. We are primarily interested in contrasting the treatment GDC plus IDC to the other three treatment options, as previous reports found GDC plus IDC to be the most effective treatment among the four.

Due to intermittent absence from the monthly visits, the missingness pattern of the monthly cocaine usage status is non-monotone. Because drug use is a $``$sensitive$"$ outcome, such that 
participants may not attend monthly visits because of their use of cocaine that would be detectable from urine toxicology screens, the missingness is thought likely to be MNAR \citep{robins1997non,vansteelandt2007estimation}. 
That is, it is thought likely that some fraction, possibly a relatively large fraction, of monthly cocaine use assessments may have been missing because participants were using cocaine. 
To address this concern,  we implement the proposed FCS-NSC algorithm and generate 20 imputed datasets based on the NSC missing data mechanism. We include the two fully-observed covariates (site and treatment group) in the FCS models; specifically, the models include the main effects of both covariates and $(Y_{-k}, M_{-k})$, in addition to interactions between treatment group and $(Y_{-k}, M_{-k})$ only. This specification is based on the assumption that parameters of the missing data mechanism are more likely to vary 
by treatment group but not site. 
For the endpoint of primary interest in this analysis, 
abstinence for at least three consecutive months during the six-month treatment, we fit a logistic regression model that includes indicator variables for the main effects of site and treatment group to each of the
multiple imputed datasets. Estimates of the contrasts between treatment groups, as well as their associated p-values, were obtained from the combined analyses using Rubin's rules \citep{rubin2004multiple}. 

The results of the treatment group comparisons of interest are presented in Table \ref{tab:app_res}. The results confirm that the GDC plus IDC has a statistically discernible higher odds of 3-month consecutive abstinence compared to any of the other three treatment options. In addition, as indicated by the estimated odds ratios in Table \ref{tab:app_res}, it appears that the other three treatment options tend to have similar odds of 3-month consecutive abstinence; these results are consistent with the results reported in \cite{crits1999psychosocial}.

\begin{table}[htbp]
  \centering
  \caption{Comparisons of GDC plus IDC to the other three treatment options. The results are based on 20 imputations, combined using Rubin's rules.}
  \vspace*{.15in}
    \begin{tabular}{r|ccc}
    \hline
    Treatment contrast      & log-OR & SE    & P-value \bigstrut[b]\\
    \hline
    GDC + SE vs. GDC + IDC & -0.964 & 0.314 & 0.002 \bigstrut[t]\\
    GDC + CT vs. GDC + IDC & -0.887 & 0.317 & 0.005 \\
    GDC vs. GDC + IDC & -0.816 & 0.318 & 0.011 \bigstrut[b]\\
    \hline
    \end{tabular}%
  \label{tab:app_res}%
\end{table}%

So far, the analysis based on the NSC assumption allows the probability that cocaine use is missing at any particular study visit to depend on (possibly unobserved) cocaine usage at any prior or subsequent visits. However, to allow the probability that cocaine use is missing at any particular study visit to depend on the (possibly unobserved) value of that variable, we must perform a sensitivity analysis.    
Specifically, we  perform a sensitivity analysis, anchored at the NSC model in the FCS-NSC algorithm, to examine how robust the above results are against potential violations of NSC. We follow the specification of sensitivity parameters as outlined in Section 2.3 and additionally assume that $\lambda_{MY}$ differs across the four treatment groups.
The rationale for allowing $\lambda_{MY}$ to differ across treatment groups is because treatment group comparisons may be more sensitive to {\em differential} MNAR missingness. 
As a result, there are four sensitivity parameters in these analyses. For the sensitivity analysis of the comparison between a specific pair of treatment groups, we vary the sensitivity parameters for these two treatments and fix the other two sensitivity parameters to be zero. We consider the range of the sensitivity parameters ($\lambda_{MY}$) to be between  0 and 2, which correspond to odds ratios between 1 and 7.4, where the odds ratio compares the odds of being missing for non-abstinent (use of cocaine) versus abstinent at a given study visit, conditional on the (possibly unobserved) outcome variables at all other study visits and their missingness indicators. This choice of range for 
$\lambda_{MY}$ was made in consultation with Dr. Weiss, whose clinical expertise and experience in conducting substance use disorder trials suggest that those with missing data may be approximately 2–3 times more likely to be using drugs. Note, we do not examine negative values of $\lambda_{MY}$ since clinical evidence suggests that non-abstinent 
study participants are much more likely to miss a study visit than abstinent study participants. That is, clinical trialists in substance use disorders have a strong prior belief that those with missing data are {\em more likely} to be using drugs.

In Figure \ref{fig:app-sen} we illustrate heatmaps of the p-values associated with the comparisons of treatment groups as we vary the sensitivity parameters. The panels in Figure 1 are colored such that a p-value larger than the conventional 0.05 significance level corresponds to blue shading, and red shading otherwise. From Figure 1, we can see that for the contrasts GDC+SE vs. GDC+IDC and GDC+CT vs. GDC+IDC, the p-values remain below 0.05 across all values of $\lambda_{MY}$, and in particular for values of $\exp(\lambda_{MY})$ 
between 2 and 3 (the range suggested by our clinical expert, Dr. Weiss). This indicates that conclusions from the earlier analysis based on the NSC assumption are robust to departures from NSC for these two particular contrasts. For the contrast of GDC vs. GDC+IDC, the p-values are greater than 0.05 only when $\lambda_{MY}$ for the GDC + IDC treatment group is extremely large ($>1.7$); in contrast, for the remainder of the parameter domain, the p-values are all less than 0.05 and the result of the main analysis under NSC is preserved. In addition, $\lambda_{MY} >1.7$ corresponds to odds ratios greater than 5.5, a magnitude of association far greater than that suggested by our clinical expert. Overall, both the main analysis and the sensitivity analysis suggest that GDC + IDC is more effective in treating cocaine use disorder than the other three treatment options. These results are also consistent with an earlier sensitivity analysis of these data focusing on a different treatment outcome, the total number of abstinent months, reported in \cite{fitzmaurice2019sensitivity}, based on a joint model that combines a MNAR selection model with a generalized linear mixed model for the binary variables.

\begin{figure}
    \centering
    \includegraphics[width=0.95\textwidth]{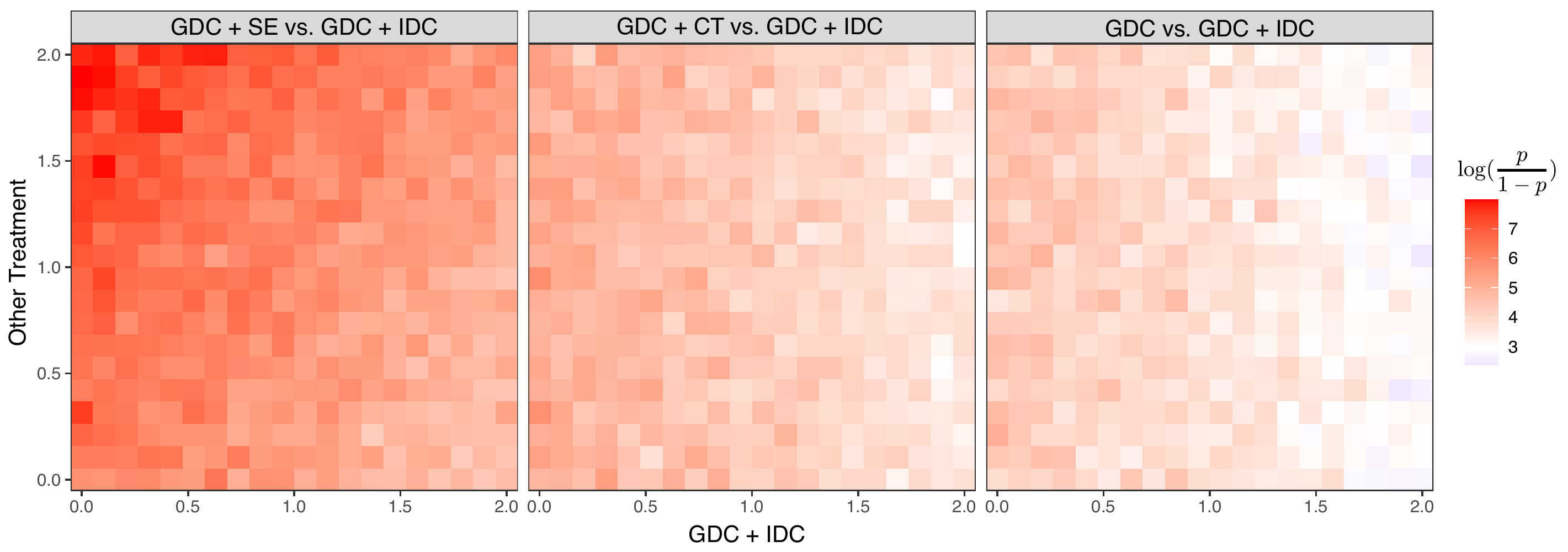}
    \caption{Results of sensitivity analysis for the comparisons between treatment groups. The p-values are derived using the same procedure as in the main analysis except where the imputations are obtained under different values of the sensitivity parameters. Colors are specified such that blue corresponds to p-values larger than 0.05 and red otherwise.}
    \label{fig:app-sen}
\end{figure}

\section{Discussion}
The missing at random (MAR) assumption is often regarded as a useful benchmark assumption for the analysis of partially missing data (e.g., see recommendations of the NRC report on handling missing data in clinical trials \citep{national2010prevention}). However, we note that the MAR assumption is far more compelling for the case of monotone patterns of missing data, as might arise in a longitudinal study with attrition or dropout. 
With non-monotone patterns of missingness, it is very difficult to conceive of MAR mechanisms that are not also MCAR \citep{robins1997non}. For example, the non-monotone MAR mechanism used in the simulations in Section 3 was somewhat contrived and does not model beliefs about the missing data process in most substantive applications; it was used only to explore the properties of the proposed NSC estimator when the true missingness mechanism is MAR. This has provided the impetus for considering a missingness model for data that are missing non-monotonically and MNAR. 

The NSC model adopted in this paper is such a non-monotone MNAR mechanism that can potentially provide an alternative to MAR as a benchmark assumption for the analysis of partially missing data. The NSC model permits a relatively high degree of interdependence between the missingness indicators and the underlying measured but sometimes missing variables. Although the missingness indicator $M_k$ is assumed to be conditionally independent of the  outcome variable
$Y_k$, NSC does allow for very general dependence on all other possibly missing outcome variables as well as all other missingness indicators.  
The results of simulation studies reported in Section 3 confirm that the proposed  approach to MI using FCS algorithms derived from a NSC model is almost unbiased in finite sample and across
a wide range of rates of missingness. Results in Section 3 also confirm that the proposed approach is
asymptotically unbiased. In contrast, the more conventional FCS-MAR approach,  
where the set of conditioned variables in the FCS models include only the other outcome variables, leads to biased estimates when the true missing
data mechanism is NSC. The bias was found to be relatively small when the missing rate is low (20\%) but increases quite dramatically when the missing rate is 40\%. 

On a related point, a general recommendation that is often cited in the statistical literature on conventional MI is to use a large number of observed variables to impute the missing data. The rationale behind this advice is that it may be more reasonable to assume missingness is MAR when conditioning on a larger number of observed variables. For example, when discussing the inclusion of missingness indicators of other variables as auxiliary covariates in the
imputation model, \cite{tompsett2018use}
recommend their inclusion by default, making the intuitively reasonable argument that their omission may potentially introduce bias, whereas a model that unnecessarily includes them  ``...should not have a negative impact on imputations''. The rationale for the latter part of this recommendation is the seemingly 
sensible principle that a FCS model that includes 
unnecessary auxiliary covariates should not introduce bias in the imputations. 
However, the results in Section 2.4 and Section 3 illustrate that this is not the case and that
there are potential perils to conditioning on the missingness indicators. Specifically, when the
true missing data mechanism is MAR, conditioning on missingness indicators of other variables in the
imputation model will, in general, introduce bias rather than reduce it.
These results highlight that the decision to 
include missingness indicators 
requires careful consideration of whether the true missing data mechanism is thought to be MAR or NSC. 

It is also important to recognize that the validity of NSC or MAR cannot be checked empirically from the data at hand (e.g., \cite{molenberghs2008every}), i.e., short of tracking
down the missing data, the NSC and MAR assumptions about missingness are wholly unverifiable from the data at hand. This has two important consequences. First, the choice between NSC and MAR as a benchmark assumption for the analysis of data missing  non-monotonically must be made on external grounds (completely unrelated to the data at hand),  say,  in consultation with a subject matter expert. As discussed earlier, we would argue that it is exceeding difficult to conceive of models for data missing non-monotonically that are MAR, e.g., in general, it is remarkably difficult to generate simulated data that would satisfy the MAR assumption  (as shown by our simulations in Section 3 with a very contrived non-monotone MAR mechanism).  Therefore, following \cite{cox2011principles}, page 104, we would argue that this makes the choice of MAR (over NSC) far less appealing and meaningful: 
$``$A test of meaningfulness of a possible model for a data-generating process is whether it can be used directly to simulate data.$"$ 
In addition, Theorem 1 in Section 2.4 also suggests that it would not be straightforward 
to define sensitivity parameters for a MAR model, as it would have some $\lambda$ terms with both $Y_k$ and $M_k$ in it.
Second, because the observed data provide no information
that can either support or refute NSC (or any other MNAR mechanism), it is important to conduct
sensitivity analysis to departures from the conditional independence assumption. As was discussed in Section 2.3, and illustrated in Section 4, our proposed FCS algorithm under NSC can easily be extended to incorporate a set of sensitivity parameters that have simple interpretations while preserving the computational efficiency of the imputation approach.
Thus, when considering a benchmark assumption for data missing non-monotonically, 
there is a natural and transparent approach to conducting sensitivity analysis to departures from NSC that relates 
inferences to one or more parameters that capture departures from the 
conditional independence assumption for $Y_k$ and
$M_k$.   
In summary, some potential benefits of the NSC model assumption for data missing non-monotonically are that (i) unlike MAR, it provides a plausible mechanism for non-monotone patterns of missingness, (ii) the conditional independence assumption is straightforward to interpret and explain to subject matter experts, (iii) the NSC model can be implemented in a straightforward way using multiple imputations based on FCS, and (iv) the NSC model assumption can be readily modified to allow for an interpretable sensitivity analysis.

In Section 2, we proposed multiple imputations based on FCS as an alternative 
to the more computationally challenging IPW/AIPW estimators for the NSC model proposed in 
\cite{malinsky2022semiparametric}. FCS is a natural approach to multiple imputation
because the NSC model assumption relates directly to the fully conditional distribution of each outcome variable. That is, to incorporate
the NSC assumption in the multiple imputations, the set of conditioned variables in the FCS models must include all other outcome variables as well as all other missingness indicators. 
Relative to the IPW estimator for the NSC model, where weights associated with different 
non-monotone patterns of missingness are estimated through the solution of 
a set of estimating equations whose dimension is $O(3^K)$, MI via FCS only requires specification of $K$ regression models for each of the $K$ outcomes, but with missingness indicators included as additional covariates. The proposed approach can also be readily implemented in existing software for MI (e.g., using the \textit{mice} package in R). In addition, unless sample size is very large relative to $K$, multiple imputation may be considered an appealing alternative to IPW 
due to the fact that the proportion of ``complete cases'' can be relatively small 
when data are missing non-monotonically, even if the ``positivity'' assumption required of IPW may be strictly met.
We note that \cite{malinsky2022semiparametric} also proposed an augmented IPW (AIPW) estimator where the augmentation term incorporates information from all the missingness patterns. However, this AIPW estimator is computationally challenging to implement when $K$ is relatively large. 
Finally, it should be acknowledged that a distinct advantage
of the IPW approach over multiple imputation is that it does not need to model the joint distribution $p(Y)$; it only requires
correct specification of the model for the non-monotone patterns of missingness, $p(M | Y)$. This makes the IPW approach particularly appealing when $K$ is relatively small so that the 
solution of the set of IPW estimating equations is not too computationally demanding.  

In applications of the NSC model for the analysis of data that are missing non-monotonically, in particular when $K$ is relatively large, it will be necessary to place some restrictions on the joint distribution $p(Y)$ and the conditional interdependence between $M_k$ and $Y_{-k}$. This is due in part to the following factors. The first is the lack of a convenient joint distribution for multivariate binary variables, unlike the multivariate normal distribution (i.e., a joint distribution that is a function of the first two moments only). Paradoxically, the multinomial joint distribution for multivariate binary data requires specification of many higher-order moments despite the fact that there is, in some sense, substantially less information in the binary than in the continuous data case. Moreover, due to the curse of dimensionality, estimation of the higher-order moments would face serious difficulties with sparseness of data once the number of binary variables exceeds 5 or 6. As a result, some degree of simplification of the NSC model is required, with the
number of restrictions driven by both $K$ and the sample size. In Sections 2$-$4 
we placed restriction on higher-order conditional moments by imposing constraints on the higher-order interaction terms in the FCS models, assuming that the corresponding regression parameters, the $\lambda$'s, are zero. For example, in the main effect NSC simulation scenario, we place the constraints that $\lambda_{Y_IM_J} = 0$ when $|I|>1$ and $|J|>1$ or $|I| + |J|>2$. This implies that the logistic model of $Y_k|Y_{-k},M_{-k}$ contains only main effects of $Y_{-k}$ and $M_{-k}$.
Relaxing constraints on $\lambda_{Y_kY_{k'}M_{k''}}$ would imply that 
the logistic model of $Y_k|Y_{-k},M_{-k}$ in also contains interaction terms $Y_kM_{k'}$.
In principle, additional constraints can be relaxed provided the data are not too sparse.
Similarly, 
when developing a sensitivity analysis strategy for departures from NSC, 
we have incorporated a fixed and known sensitivity parameter, $\lambda_{M_kY_k}$ for the conditional 
dependence of $Y_k$ on $M_k$, thereby directly relating missingness to the possibly unobserved 
binary variable. The choice of values for $\lambda_{M_kY_k}$ (or for
$\lambda_{MY}$, when we set $\lambda_{M_kY_k} = \lambda_{MY}$ for $k=1,\ldots,K$) should be made on subject-matter grounds guided by knowledgeable study investigators.
The rationale for our choice of this very simple model for the conditional
dependence of $Y_k$ on $M_k$ is two-fold. First, it is advantageous to keep the dimension of the sensitivity parameters low. With $K$ variables 
there are potentially an excessively large 
number of parameters for describing the possible conditional
dependence of $Y_k$ on $M_k$ and sensitivity analysis over such a large number of parameters is simply not viable. Second, the parameter $\lambda_{M_kY_k}$ (or
$\lambda_{MY}$) is straightforward to interpret and makes transparent what is assumed about the missing variables. 

Finally, we note that imputation models that also include the indicators of missingness have been suggested in the literature by \cite{tompsett2018use}, and more recently by \cite{beesley2021multiple}. In particular, in the context of sensitivity analysis, \cite{tompsett2018use}   recommend FCS imputation models also include the indicators of missingness.  However, no formal statistical basis for including indicators of missingness was given in \cite{tompsett2018use}. \cite{beesley2021multiple} propose a set of identifying assumptions that are stronger than NSC and, using Taylor series and other approximations, lead to FCS imputation models that also include dependence on the indicators of missingness.
In the previous sections we have shown that the NSC assumption provides a formal statistical justification for including all other missingness indicators in the FCS model and that these conditional distributions are identifiable given the observed data. Additionally, the ``self censoring'' parameters (e.g, $\lambda_{M_kY_k}$) for the conditional dependence between $Y_k$ and $M_k$ are not identifiable given the observed data but are the natural parameters for sensitivity analysis to departures from NSC.

In closing, some possible avenues of future research include considering whether more flexible 
FCS models, e.g., general machine learning algorithms as alternatives to logistic regressions, can protect against
potential model misspecification in the FCS algorithm. 
Another avenue for future research is the extension of the FCS-NSC algorithm to data that are non-binary. In principle, the extension to continuous data, count data, ordinal data, and so on, is relatively
straightforward once a suitable choice of link function is chosen in the model for $Y_k|Y_{-k},M_{-k}$; alternatively, relevant machine learning algorithms could be explored.

\section*{Acknowledgements}
The authors thank Dr. Paul Crits-Christoph for sharing data from the NIDA Collaborative Cocaine Treatment Study. 

\section*{Declaration of conflicting interests} 
Roger Weiss has consulted to Alkermes. 
All other authors report no conflicts of interest. 

\section*{Funding} 
This work was supported by the National Institutes of Health, National Institute on Drug Abuse grants [NIDA R33 DA042847, UG1 DA015831].

\vspace*{.5in}


\bibliographystyle{chicago}
\bibliography{reference}

\end{document}